\documentclass[a4paper,USenglish,cleveref, autoref, thm-restate]{lipics-v2021}

\pdfoutput=1 
\hideLIPIcs  

\title{Minimizing Vertical Length in Linked Bar Charts} 

\author{Steven {van den Broek}}{TU Eindhoven, the Netherlands }{s.w.v.d.broek@tue.nl}{https://orcid.org/0009-0005-6677-3916}{}

\author{Marc van Kreveld}{Utrecht University, the Netherlands}{m.j.vankreveld@uu.nl}{https://orcid.org/0000-0001-8208-3468}{}

\author{Wouter Meulemans}{TU Eindhoven, the Netherlands }{w.meulemans@tue.nl}{https://orcid.org/0000-0002-4978-3400}{Partially supported by the Dutch Research Council (NWO) under project number VI.Vidi.223.137.}

\author{Arjen Simons}{TU Eindhoven, the Netherlands }{a.simons1@tue.nl}{https://orcid.org/0009-0008-1271-180X}{Supported by the Dutch Research Council (NWO) under project number VI.Vidi.223.137.}

\authorrunning{S. van den Broek, M. van Kreveld, W. Meulemans and A. Simons} 

\Copyright{Steven van den Broek, Marc van Kreveld, Wouter Meulemans and Arjen Simons} 



\category{} 

\relatedversion{} 

\acknowledgements{Research on the topic of this paper was initiated at the 8th Workshop on Applied Geometric Algorithms (AGA 2024) in Otterlo,
The Netherlands.}

\nolinenumbers 

\newcommand{\Reals}{\mathbb{R}}

\EventEditors{John Q. Open and Joan R. Access}
\EventNoEds{2}
\EventLongTitle{42nd Conference on Very Important Topics (CVIT 2016)}
\EventShortTitle{CVIT 2016}
\EventAcronym{CVIT}
\EventYear{2016}
\EventDate{December 24--27, 2016}
\EventLocation{Little Whinging, United Kingdom}
\EventLogo{}
\SeriesVolume{42}
\ArticleNo{23}

\begin{document}

\maketitle

\begin{abstract}
A linked bar chart is the augmentation of a traditional bar chart where each bar is partitioned into blocks and pairs of blocks are linked using orthogonal lines that pass over intermediate bars. 
The order of the blocks readily influences the legibility of the links. 
We study the algorithmic problem of minimizing the vertical length of these links, for a fixed bar order. The main challenge lies with \emph{dependent} links, whose vertical link length cannot be optimized independently per bar. We show that, if the dependent links form a forest, the problem can be solved in $O(nm)$ time, for $n$ bars and $m$ links. If the dependent links between non-adjacent bars form a forest, the problem admits an $O(n^4m)$-time algorithm. Finally, we show that the general case is fixed-parameter tractable in the maximum number of links that are connected to one bar.
\end{abstract}

\section{Introduction}
\label{sec:intro}
Bar charts are a ubiquitous tool for visualizing scalar values across categories. 
Stacked bar charts, in particular, allow different quantities to be aggregated in a single column. In their traditional form, they primarily show \emph{single-category} values: values that are uniquely attributable to a specific category. 
In many settings, however, certain quantities are not uniquely attributable to a single category but instead relate to multiple categories. 

Such \emph{cross-category} values may arise in different forms. They may represent shared quantities, for example, in a bar chart encoding total communication per account: the communication between two accounts is present in both bars. They may also represent pairwise uncertainties, that is, quantities that may belong to one of two categories. For example, in election poll results, groups of voters may hesitate between two political parties, or in the analysis of pollution, factories near country borders may contribute to the pollution of either country, though the exact distribution may not be known.

As standard bar charts cannot directly visualize such cross-category values, \emph{linked bar charts} were recently introduced \cite{PairWiseUncertainty}: the single- and cross-category values partition each bar into \emph{blocks} of appropriate height, such that the total bar height reflects the aggregate value of the category. Each cross-category value is then visualized through a \emph{link}: a polyline between the blocks that passes over intermediate bars (\autoref{fig:pairwise-uncertainty}). We refer to blocks as \emph{unlinked} or \emph{linked} blocks, for single-category and cross-category values, respectively, that is, depending on whether the block is linked to another.

\begin{figure}[t]
    \centering
    \includegraphics[]{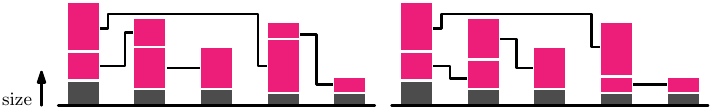}
    \caption{Two linked bar charts  \cite{PairWiseUncertainty} that show the same data using different vertical orderings, with cross-category scalar values between linked blocks (pink) and single-category values drawn as unlinked blocks (gray).
    } 
    \label{fig:pairwise-uncertainty}
    \vspace{-1\baselineskip}
\end{figure}

Note that linked bar charts effectively show a weighted graph: the bars are vertices, weighted by their single-category values, and linked blocks are edges, weighted by their cross-category values.
The quality of the resulting drawing readily depends on the order of the bars, as well as the order in which the blocks are stacked on top of each other. Arising from (orthogonal) graph-drawing literature, there are various natural measures \cite{di1999graph}, such as the number of crossings \cite{crossingsandplanarization,purchase1996validating,radermacher2018geometric}, the length of the links \cite{tamassia1987embedding,tamassia1988automatic} and the number of bends \cite{blasius2016optimal,formann1993drawing,purchase1996validating,tamassia1987embedding}. Vertical distance between elements has also been considered as a quality measure for other visualizations, such as storylines \cite{Storylines} and parallel coordinate plots \cite{parallelCoordinates}.

For quality measures that rely solely on the bar order without considering the stacking order of the blocks, the problem is effectively that of drawing the (unweighted) graph with a one-page book embedding. 
The outerplanar graphs are exactly the family of graphs that can be drawn without edge crossings in this style \cite{BookThicknes}. For such graphs, minimizing total and maximum (horizontal) edge length, cutwidth, or bandwidth is polynomial-time solvable \cite{PlanarLinearArrangements}. 
These problems are NP-hard for general graphs \cite{NP-Completeness}.

In contrast, the stacking order of the blocks adds a new dimension to this classic graph-drawing problem, which remains unstudied at this time.

\subparagraph{Contributions.} We study the following problem: given a weighted graph with fixed vertex order, we aim to compute a linked bar chart---a stacking of blocks in each bar---that minimizes the total vertical link length, while avoiding crossings between links that share an endpoint.
In \autoref{sec:defs}, we introduce our definitions. We distinguish between \emph{dependent} and \emph{independent} links, indicating whether the optimal position of a linked block depends on the position of the block it is linked to.
In \autoref{sec:forest}, we describe an $O(nm)$-time algorithm, for cases with $n$ bars and $m$ links, where the subgraph of dependent links is a forest.
If the subgraph of dependent links between non-adjacent bars is a forest, the problem can be solved in $O(n^4 m)$ time (\autoref{sec:forestplus}).
The problem is fixed-parameter-tractable, parameterized by the maximum degree of a bar (\autoref{sec:fpt}). This algorithm still runs in polynomial time, if only the degree of the subgraph of dependent links is bounded by a constant.

\section{Definitions, notation and observations}
\label{sec:defs}

Our input is a weighted graph $G = (V,E, w)$, where $V$ represents a sequence of $n$ vertices $v_1,\ldots,v_n$ in fixed order and $E \subseteq V^2$ a set of $m$ edges. The weight function $w: V \cup E\rightarrow \Reals^+$ assigns weights both to vertices (single-category values) and edges (cross-category values).
The \emph{span} of an edge is the subsequence of vertices in $V$ between its endpoints, including these endpoints. 
The \emph{intermediate} vertices of an edge refer to its span minus its endpoints.

\subparagraph{Bars, blocks and links.} 
To visualize $G$, we draw the vertices in $V$ as a sequence of unit-width \emph{bars} $B_1,\hdots,B_n$, arranged on a horizontal baseline with a spacing of one unit between adjacent bars.
Note that the horizontal position of the bars is fixed by their order in the sequence $B_1, \dots, B_n$. 
Each vertex $v_i$ corresponds to bar $B_i$ for all $1 \leq i \leq n$; we use bar to refer to the vertex in $V$ as well as its visual representation, as there is a one-to-one mapping between the two.

An edge $e$ induces a \emph{linked block} in each of the two bars it connects: a rectangle of unit width and height $w(e)$. A vertex $v$ induces an \emph{unlinked block} of height $w(v)$ in its corresponding bar.
The height of a block $b$ is denoted by $h(b)$.

We draw an edge $e$ as an orthogonal \emph{link} 
that connects the two linked blocks corresponding to $e$; see also \autoref{fig:pairwise-uncertainty}. 
A link connects the centers of two blocks; we refer to these centers as the endpoints of the link. 
If at least one endpoint is placed higher than all intermediate bars of $e$, then the link generally has two bends; a link may be horizontal in some cases and have zero bends. 
Otherwise, the link is drawn using four bends. Links always pass over all intermediate bars. 
We assume that between two consecutive bars, there is enough horizontal space for the links to run between those bars, and ignore this issue from now on.

\subparagraph{Stacking blocks.}
As the bar order is fixed, the horizontal length of links is as well. 
To minimize total vertical length, we stack the blocks within a bar appropriately while avoiding crossings between links that connect to the same bar. A stacking of blocks corresponds to an order on the incident edges of a bar. Unlinked blocks are placed at the baseline, below the stacking of linked blocks.

Consider a bar $B_j$. 
All incident edges $\{ B_i, B_j \}$ to an earlier bar $B_i$ go towards the left from~$B_j$. 
We denote this ordered sequence of edges by $L_j = \{l_1,\ldots,\}$, sorted in decreasing order of index of the other bar. 
Similarly, the rightward edges are denoted by $R_j = \{r_1,\ldots, \}$---the edges $\{B_j,B_k\}$ with $k > j$---in increasing order of index of the other bar. 
To avoid crossings between links that originate from a common bar, the stacking order must contain both $L_j$ and $R_j$ as a subsequence. 
That is, the stacking order must be some merge of these two ordered sequences.
Any merge is valid, but they incur different vertical link lengths, as shown in \autoref{fig:pairwise-uncertainty}.
For ease of notation, we identify an edge $e$ in $L_j$ or $R_j$ 
with the corresponding block in bar $B_j$.

Merging left- and rightward subsequences in this way implies that the number of crossings in our visualization equals that in the crossings in the one-page book drawing of $G$ that respects the fixed vertex order.
Two links intersect if and only if their spans intersect in at least two bars and neither is a subset of the other.

Consider some linked block $b$ in a bar $B_j$, corresponding to vertex $v_j$, for edge $l_i \in L_j$.
Its vertical center point is determined by its height, the height of the unlinked block, and all blocks prior to it in $L_j$---the order of the blocks in $L_j$ is fixed---and also by all blocks in $R_j$ that are chosen to be below it in the stacking order.
Hence, the center coordinate of a block is influenced only by the number of blocks in the other sequence occurring before it.
With $k$ blocks of $R_j$ below $b$, the vertical center of $b$ is given by $y(b,k) = \frac{1}{2} h(b) + w(v_j) + \sum_{x = 1}^{i-1} h(l_x) + \sum_{x=1}^{k} h(r_x)$.
As the $y$-position of a block is uniquely determined by $k$, we also refer to $k$ as the position of the block.
The situation is symmetrical for a block in $R_j$.

For the remainder, we do not explicitly treat unlinked blocks: their fixed effect is readily captured via the $y$-function defined above. Hence, we refer to linked blocks simply as blocks. 

\subparagraph{Link types.}
Consider an edge $e = \{ B_i, B_j \}$ with $i < j$. Its linked blocks $b$ and $b'$ must be in $R_i$ and~$L_j$. Let $\uparrow\!b = y(b, |L_i|)$ and $\downarrow\!b = y(b, 0)$ be the highest and lowest possible $y$-coordinate for $b$ its center and analogously define~$\uparrow\!b' = y(b', |R_j|)$ and $\downarrow\!b' = y(b', 0)$. Let $H$ 
denote the height of the tallest intermediate bar.

We call a link between block $b$ and $b'$ \emph{independent}, if its vertical length can be minimized by placing its blocks closer to a fixed target $t$ that does not depend on the relative position of $b$ and $b'$. That is, its vertical length can be expressed as $|t - y| + |t - y'|$, for given center coordinates $y$ for $b$ and $y'$ for $b'$, for some fixed target $t$. 
In the following cases, links are independent 
(see \autoref{fig:independent-link}, left): 

\begin{enumerate}
    \item If there is an intermediate bar that is higher than the highest coordinate for one of the blocks, the link must always go up to $H$ from the lower block: the target is $H$. Formally, if $H \geq\; \uparrow\!b$ or $H \geq\; \uparrow\!b'$, then $t = H$.
    \item Otherwise, if the intervals are interior disjoint, the direction of the vertical piece is always the same: we can use as a target the lowest position of the higher interval. Formally, if $H < \min\{ \uparrow\!b, \uparrow\!b'\}$ and $\downarrow\!b \geq\; \uparrow\!b'$, then $t = \downarrow\!b$; analogous for $\uparrow\!b \leq\; \downarrow\!b'$.
    \item Otherwise, if one of the intervals is a singleton, we use the fixed value as a target. Formally, if $\uparrow\!b = \downarrow\!b$ (that is, $L_i = \emptyset$), then $t = \uparrow\!b$; analogous for $\uparrow\!b' = \downarrow\!b'$ (that is, $R_j = \emptyset$).
\end{enumerate}

\begin{figure}[b]
    \centering
    \includegraphics{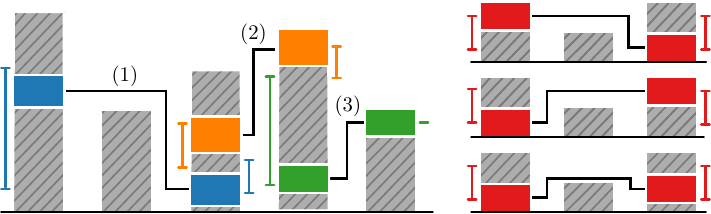}
    \caption{Schematics of links, as determined by the interval of possible center coordinates (drawn ranges). Left: the three independent link cases.
    Right: three possible placements for a single dependent link.}
    \label{fig:independent-link}
\end{figure}

If none of the cases hold, the link is \emph{dependent}: its vertical length depends on the relative position of the two blocks, and we cannot generally assign a static target (see \autoref{fig:independent-link}, right). 
Dependent links come in two variants: \emph{adjacent dependent links} and \emph{non-adjacent dependent links}, depending on whether the two linked bars are consecutive.

We abbreviate the various link types as IL (independent link), DL (dependent link), ADL (adjacent dependent link) and NADL (non-adjacent dependent link). 
We refer to edges and blocks also via these types: for example, a dependent edge is an edge with a dependent link.

Each of the link types induces a subgraph of $G$. In particular, our results rely on the structure of the subgraph of dependent links.

\begin{restatable}{lemma}{noDLcross}\label{obs:noDLcross}
The subgraph of dependent links is a plane one-page book embedding for the given bar order, and thus is outerplanar. 
\end{restatable}
\begin{proof}
By definition, a dependent edge $e$ requires that its intermediate bars have strictly lower height than both of its endpoints. 
Hence, edges between an intermediate bar of $e$ and a bar outside the span of $e$ cannot satisfy this property. 
Therefore, links of dependent edges cannot cross: their spans are either disjoint, one is contained in the other, or they share one endpoint $B_i$, with one block being in $L_i$ and the other in~$R_i$.

For the fixed bar order, the dependent links of $G$ form a one-page book embedding without crossings. Hence, it is outerplanar \cite{BookThicknes}.
\end{proof}

\subparagraph{Computing with links.}
By their definition, the contribution of independent links to the total vertical link length can be decomposed into a \emph{cost} for placing each of its blocks, being the distance of its placement to the target of the link. For dependent links, their \emph{cost}---contribution to the total vertical link length---can be computed only when knowing the placement of both its blocks. We use $\lambda$ to denote these costs. That is, an independent block $b$ placed at position $i$ has cost $\lambda(b,i) = |t - y(b,i)|$, where $t$ denotes its target. For a dependent link $e = (b,b')$, the cost depends on the the largest intermediate bar with height $H$: if both endpoints lie below $H$, it is the sum of the two vertical segments; otherwise, it is the absolute difference of the endpoints (see \autoref{fig:independent-link}, right). Hence, letting $i$ and $j$ denote the positions of the endpoints, its cost is computed as:
\[ \lambda(e,i,j) = \left\{ 
\begin{array}{ll}
|H - y(b,i)| + |H - y(b',j)| & \text{if }  H > y(b, i)  \text{ and } H > y(b',j) \\
|y(b,i) - y(b',j)|  & \text{otherwise}
\end{array}
\right.
\]

We assume that we know, for each link, its type and the maximum height of its intermediate bars, such that the cost of a placed independent block and of a placed dependent link can be computed in $O(1)$ time. By the lemma below, we can precompute such information efficiently, and the running time of this precomputation step is dominated by all our other algorithms.

\begin{restatable}{lemma}{computationLinkTypes}\label{lem:computationLinkTypes}
Given $G = (V, E, w)$ as adjacency list with only $V$ in order, we can compute the $L$- and $R$-sets for all bars in $O(n + m)$ time, and augment the bars and links in $O(\min\{n^2, n+ m\log n\})$ time, such that the type of a link and $\lambda$ can be computed in $O(1)$ time.
\end{restatable}
\begin{proof}
Consider the bars that represent the vertices in $V$. To compute the $L$- and $R$-sets, we use the following algorithm, as the bar order matches the order of $L$ (reversed) and $R$. So, we can initialize the lists to be empty, and traverse the bars in order. For a bar $B_i$, we iterate over the edges incident to $V_i$. Let $V_j$ denote the other endpoint of such an edge. If $j < i$, then we place the block of the link corresponding to the edge at the end of $R_j$: all edges to sets before $V_i$ should come before and have already been added. If $j > i$, then we place the block of the link corresponding to the edge at the start of $L_j$: all blocks corresponding to links that come from earlier bars should come later in $L_j$ and have already been added. 
As we process each bar and edge once, this takes $O(n+m)$ time.

Assuming the $L$- and $R$-sets have been computed, we now compute the link types: first, we compute the prefix sum of $h$ for all $L$- and $R$-sets. This allows us to compute $\uparrow\!b$ and $\downarrow\!b$ for a block $b$ as well as the total height of a bar in constant time.
For each bar $B_i$, we iterate over $R_i$ in order, to determine the link type of those links. We simultaneously traverse the bars to maintain the maximum height $H$ of intermediate bars. By comparing $\uparrow\!b$, $\downarrow\!b$ of both blocks and $H$, we can determine the link type in constant time.
As we sweep over all pairs of bars, this takes $O(n^2 + m) = O(n^2)$ time in total.

When the graph has $o(n^2)$ edges, we can do slightly better. We first build a binary tree on the ordered list of bars, augmenting each node with the maximum height of a bar in the subtree. This can easily be computed in $O(n)$ time, as determining the height of a bar takes $O(1)$ time.
Now, for each edge in $G$, we can compute $H$ in $O(\log n)$ time, and use it to determine the link type in $O(1)$ time. The total running time is hence $O(n + m\log n)$.

With the above process, we have augmented every bar with the prefix sums of the $L$- and $R$-sets, and every link with its type and maximum height of an intermediate bar. Since the prefix sums allow us to compute the $y$-coordinate of a placement in $O(1)$ time and $H$ is precomputed, we can compute $\lambda(e,i,j)$ in $O(1)$ time for any link $e$ with placement $i,j$.
\end{proof}

Finally, we observe that any interactions in how blocks are stacked in a bar arise from dependent links. Thus, we can minimize the cost for each connected component in the dependent-link subgraph separately, and sum their results to minimize the total vertical link length. The algorithms presented in the next sections hence look at ways of further decomposing the dependent-link subgraph into smaller parts and parameterizing dependencies between these parts.

\section{The dependent links form a forest}
\label{sec:forest}

When the subgraph of DLs is a forest, we can minimize the vertical link length efficiently, by using the trees to determine an order in which to process the bars.
Note that this case generalizes charts in which the bars are sorted by height.

\begin{restatable}{theorem}{forestDL}\label{thm:forestDL}
Given $G = (V, E, w)$ where the subgraph of dependent links is a forest, minimizing the vertical link length takes $O(nm)$ time.
\end{restatable}

\begin{restatable}{observation}{singlemax}\label{obs:singlemax}
If the bars have a single local maximum in height, the dependent links form a collection of paths.
\end{restatable}
\begin{proof}
An NADL requires that all intermediate bars are strictly lower than its endpoints: with a single local maximum, there are only ILs and ADLs.
\end{proof}

Before we prove \autoref{thm:forestDL}, we turn our focus to a simpler case, where there are no dependent edges.

\begin{restatable}{lemma}{noDL}\label{lem:noDL}
Given $G = (V, E, w)$ without dependent edges, minimizing the vertical link length takes $O(nm)$ time.
\end{restatable}
\begin{proof}
Without dependent links, we can solve each bar in isolation, via dynamic programming.
Let $D(p,q)$ denote the minimum cost of stacking blocks $\{l_1, \ldots,l_p\} \subseteq L$ and blocks $\{r_1,\ldots,r_q\} \subseteq R$. 
Since either $l_p$ or $r_q$ must be the top block in this subproblem, we can express $D$ recursively:
\[
   D(p,q) = \left\{\begin{array}{ll}
      0 & \textrm{if }p = 0 \wedge q = 0 \\
       D(p-1,q) + \lambda(l_p,q) & \textrm{if }p > 0 \wedge q = 0 \\
       D(p,q-1) + \lambda(l_q,p) & \textrm{if }p = 0 \wedge q > 0 \\
       \min\{ D(p-1,q) + \lambda(l_p,q), D(p,q-1) + \lambda(l_q,p) \} & \textrm{otherwise} \\ 
   \end{array}\right.
\]

The final solution is then $D(|L|, |R|)$. 
In the appropriate order, that is, for increasing values of $p$ and $q$, we can determine the value in constant time per cell of $D$. 
So, solving one bar of degree $d_i$ takes $O(|L||R|) = O(d_i^2)$ time. 
Summing over all bars gives $O(\sum_i d_i^2)$ total time. Since $d_i \leq n$, we find that $\sum_i d_i^2 \leq \sum_i n \cdot d_i = n \sum_i d_i$. 
As $\sum_i d_i = 2m$ (hand-shaking lemma), we find that the total running time is bounded by $O(nm)$.
\end{proof}

In the remainder of this section, we prove \autoref{thm:forestDL}.
First, note that the trees in the forest can be optimized independently, as there are no dependent links between them.
We describe below the algorithm for a single tree $T$. Note that bars without dependent links form trees with a single vertex and their dynamic program reduces to that of \autoref{lem:noDL}.

\subparagraph{Setup.}
We root tree $T$ at an arbitrary bar; each bar $B$ in $T$ is now the root of a subtree $T_{B}$.
A dependent block in $B$ has an \emph{associated subtree}: subtree $T_{B'}$ when $B$ is part of link $\{B, B'\}$.
Without loss of generality, let $l_{p^*} \in L$ be the block in $B$ that corresponds to the dependent link that connects to the parent of $B$ in $T$. We refer to this link as the \emph{parent link}.

Let ${\downarrow_B}(p, q)$ denote the first $p < p^*$ and $q$ blocks from the $L$ and $R$ sets of $B$ respectively.
Let ${\Downarrow_B}(p, q)$ denote the blocks in ${\downarrow_B}(p, q)$ together with all blocks in subtrees associated with dependent blocks in ${\downarrow_B}(p, q)$.
We define the blocks in ${\uparrow_B}(p, q)$ and ${\Uparrow_B}(p, q)$ symmetrically, for $p > p^*$.
We define the cost of a set of blocks $S$ (such as ${\uparrow_B}(p, q)$ and ${\Uparrow_B}(p, q)$) as the cost of all independent blocks plus the cost of all dependent links whose endpoint blocks are both in the set $S$.

\subparagraph{Splitting the cost.}
We process $T$ using a post-order traversal to compute for each bar $B$ the minimum cost $P(B, k)$ of the blocks in subtree $T_B$, given that block $l_{p^*}$ has $k$ blocks of $R$ below it (\autoref{fig:dependent-subgraph-forest-app}, left).
To compute the $P(B,k)$ values efficiently, we observe that placing $l_{p^*}$ at position $k$ splits $B$ into two parts: ${\downarrow_B}(p^* - 1, k)$ and ${\uparrow_B}(p^* + 1, k + 1)$ (\autoref{fig:dependent-subgraph-forest-app}, middle).
Let $D_\downarrow(p,q)$ and $D_\uparrow(p,q)$ denote the minimum cost of ${\Downarrow_B}(p, q)$ and ${\Uparrow_B}(p, q)$ respectively (\autoref{fig:dependent-subgraph-forest-app}, right).
Then $P(B,k) = D_\downarrow(p^*-1,k) + D_\uparrow(p^*+1,k+1)$; the only bar without a parent link is the root $B_\textrm{r}$, which is simply computed as $P(B_\textrm{r}) = D_\downarrow(|L_\textrm{r}|,|R_\textrm{r}|)$, reflecting the optimal cost for the entire tree.

\begin{figure}[t]
    \centering
    \includegraphics[page=3]{forest.pdf}
    \hfill
    \includegraphics[page=1]{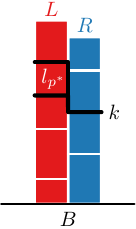}
    \hfill
    \includegraphics[page=2]{forest.pdf}
    \caption{
    Left: $P(B, k)$ is the minimum total cost of all blocks in $B$ and its child subtrees, except that of block $l_{p^*}$ (hatched), when $l_{p^*}$ is placed with $k$ blocks of $R$ under it. Grey dots are independent links. 
    Middle: block $l_{p^*}$ connects $B$ to its parent in $T$. When placing $k$ blocks of $R$ under $l_{p^*}$, $B$ is split (black line) into two parts that can be stacked independently. Right: $D_\downarrow(p, q)$ is the minimum cost of the first $p$ and $q$ blocks from $L$ and $R$, respectively. }
    \label{fig:dependent-subgraph-forest-app}
\end{figure}

\subparagraph{Computing each part.}
We can compute both $D$ tables using a dynamic program akin to that of \autoref{lem:noDL}. Specifically, $D_\downarrow$ follows the same recursion, and $D_\uparrow$ requires a simple inversion, where we consider which of the two options of the considered sets is the lowest block instead of the highest block. However, we need to explicitly consider the dependent links in these blocks. Consider a dependent link $e$ between blocks $\{b,b'\}$; assume $b \in R$, the other case is analogous. Its cost is $\lambda(e,i,j)$ when placing $b$ at $y(b,i)$ and $b'$ at $y(b',j)$. Hence, the cost of placing block $b$ at $i$ becomes $\min_{j} \lambda(e,i,j) + P(B',j)$, where $B'$ is the associated child bar of $B$ in~$T$.

\subparagraph{Proving the theorem.} 
For each child bar, we can precompute these costs for dependent blocks in $O(d d')$ time, where $d$ and $d'$ are the degrees of bars $B$ and $B'$, respectively. Afterwards, a constant-time lookup to get the cost of a block suffices. 
Thus, we compute both $D$ tables in $O(d^2)$ time for a bar of degree $d$, resulting in a total running time of $O(\sum d^2 + \sum d \cdot d')$. Via an analogous derivation as in the proof of \autoref{lem:noDL}, this reduces to an upper bound of $O(nm)$ time.

\section{The non-adjacent dependent links form a forest}
\label{sec:forestplus}

We take our previous results one step further by only requiring that the subgraph of the NADLs is a forest. That is, there may be ADLs in addition to this forest. Note that this case is also implied by a specific structure on the total bar heights.

\begin{restatable}{theorem}{forestplusDL}\label{lem:forestplusDL}
Given $G = (V, E, w)$ where the subgraph of non-adjacent dependent links is a forest, minimizing the vertical link length takes $O(n^4m)$ time.
\end{restatable}

\begin{restatable}{observation}{singlemin}\label{obs:singlemin}
If the bars have a single local minimum in height, the non-adjacent dependent links form a forest.
\end{restatable}
\begin{proof}
The NADLs of a bar are \emph{one-sided}: they are all in $L$ or all in $R$. Assume for contradiction that a simple cycle of NADLs exists. Consider the leftmost bar $B_i$ in this cycle. It links in the cycle to $B_j$ and $B_k$ with $i < j < k$.  The next link in the cycle at $B_j$ must be in $L_j$, due to the bar being one-sided, and thus go to a bar $B_x$ with $x < j$. If $x < i$, then $B_i$ was not the leftmost bar; if $x = i$, the cycle is not simple; if $x > i$, then any path from $B_x$ to $B_k$ must cross the edge $\{B_i, B_k\}$, contradicting planarity (\autoref{obs:noDLcross}).
\end{proof}

To prove \autoref{lem:forestplusDL}, we develop an extended dynamic program, akin to that of \autoref{thm:forestDL}. We first extend the forest of NADLs to a forest $F$ that contains all NADLs and a maximal set of ADLs; that is, ADLs are added until it is no longer possible to add one without creating a cycle.
As $F$ has no DLs between different trees, we can optimize each tree independently. In the remainder, we consider a single such tree $T \in F$.

\subparagraph{Setup.}
We root $T$ at its leftmost bar and denote the subtree at a bar $B$ with $T_B$.
As in \autoref{thm:forestDL}, we associate to a block in $B$ the subtree $T_{B'}$ of its child bar $B'$.
The dependent links from $B$ to its children form subsequences of $L$ and $R$ that we denote by $L^D = ( l_1^D, l_2^D, \ldots )$ and $R^D =  ( r_1^D, r_2^D, \ldots )$.
The left of \autoref{fig:dependent-non-adjacent-subgraph-forest-app} illustrates these links and upcoming definitions.

\begin{figure}[b]
    \includegraphics[page=4]{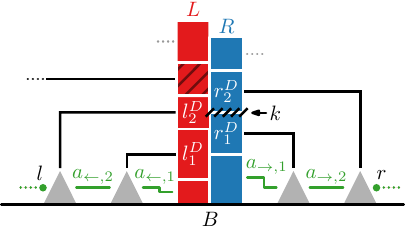}
    \hfill
    \includegraphics[page=9]{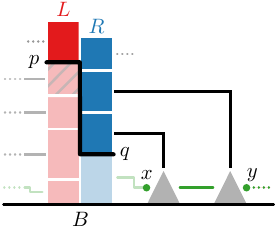}
    \caption{Left: illustration of $P(B, k, l, r)$ and ADLs $a_{\leftarrow, 1}, a_{\leftarrow, 2}, \dots$. The parent link (hatched) is excluded. Right: illustration of $D_\uparrow(p, q, x, y)$. There cannot be a leftward DL to a child of $B$. 
    }
    \label{fig:dependent-non-adjacent-subgraph-forest-app}
\end{figure}

By planarity and choice of root, the parent link of a bar $B$ is the highest dependent link in its $L$- or $R$-set. Below, we assume that it is a leftward link; the other case is analogous. 

Now, $T_B$ and the parent of $B$ have the following order in the linked bar chart, from left to right: the parent bar of $B$, then the subtrees associated with $(\dots, l_2^D, l_1^D$), then bar $B$, and lastly the subtrees associated with $(r_1^D, r_2^D, \dots)$.
An ADL can connect only bars that are adjacent in the linked bar chart.
Let $\dots, a_{\leftarrow, 2}, a_{\leftarrow, 1}, a_{\rightarrow, 1}, a_{\rightarrow, 2}, \dots$ denote the ADLs that are not in $F$ and have one endpoint in a child subtree of $B$ and one endpoint in some other bar outside that subtree, in left-to-right order. 
The arrow in the subscript denotes on which side of $B$ the link lies.

Recall that the cost of a set of blocks is the cost associated with all independent blocks and the cost of all dependent links whose endpoint blocks are both in the set.
Consider the subtree $T_B$ rooted at bar $B$. We aim to develop a dynamic program to compute the cost of blocks in $T_B$.
There are, at most, three dependent links that have only one endpoint in $T_B$: the parent link, a leftward ADL at the leftmost bar in $T_B$, and a rightward ADL at the rightmost bar in $T_B$. By parameterizing their position by $k$, $l$ and $r$, we derive a dynamic program $P(B,k,l,r)$: the minimum total cost for subtree $T_B$, with the three dependent links positioned according to the three parameters. As in \autoref{thm:forestDL}, the root has no such parameters and $P(B)$ is the optimal solution of the entire tree.

\subparagraph{Splitting the cost.}
Our approach to computing $P(B, k, l, r)$ efficiently is structurally the same as that of \autoref{thm:forestDL}.
Again, let $l_{p^*}$ denote the block in $L$ that connects to the parent of $B$, and define ${\Downarrow_B}(p, q)$ and ${\Uparrow_B}(p, q)$ as in \autoref{thm:forestDL}.
In addition to the parent link, there are at most two dependent links that have only one endpoint in ${\Downarrow_B}(p, q)$: a leftward ADL at the leftmost bar in ${\Downarrow_B}(p, q)$, and a rightward ADL at the rightmost bar.
Let $D_\downarrow(p, q, x, y)$ denote the minimum cost of ${\Downarrow_B}(p, q)$ when the at most two ADLs with only one endpoint in ${\Downarrow_B}(p, q)$ are positioned according to $x$ and $y$ (\autoref{fig:dependent-non-adjacent-subgraph-forest:2:app}, left).
We symmetrically define  $D_\uparrow(p, q, x, y)$ (\autoref{fig:dependent-non-adjacent-subgraph-forest-app}, right), noting that there cannot be dependent links above parent link $l_{p^*}$ in $L$, by the choice of root.
Value $P(B, k, l, r)$ can be derived by combining the minimum costs of the two parts of $B$ and aligning the at most one ADL that has an endpoint in each of the two parts. More precisely, the value is equal to $\min_{a, b} ( D_\downarrow(p^*-1, k, l, a) + D_\uparrow(p^*+1, k+1, b, r) )$.

\begin{figure}[b]
    \includegraphics[page=6]{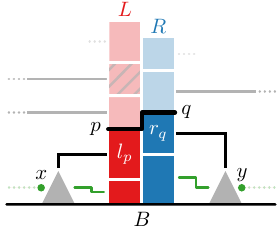}
    \hfill
    \includegraphics[page=7]{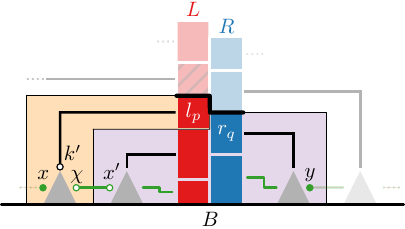}
    \caption{Left: illustration of $D_\downarrow(p, q, x, y)$. Right: the subproblems for computing $D_\downarrow(p, q, x, y)$ when dependent block $l_p$ is on top, with Cost~\ref{eq:D'} (purple) and Cost \ref{eq:P'} (orange).}
    \label{fig:dependent-non-adjacent-subgraph-forest:2:app}
\end{figure}

\subparagraph{Computing each part.}
As with $D$ of \autoref{lem:noDL} and $D_\downarrow$ of \autoref{thm:forestDL}, we can define $D_\downarrow(p, q, x, y)$ recursively by observing that either $l_p$ or $r_q$ must be the top block.
Without loss of generality, assume $l_p$ is the top block; the other case is symmetric.
Furthermore, assume w.l.o.g.\ that every DL $l_i^D$ has a corresponding ADL $a_{\leftarrow, i}$; if in the following such an ADL does not exist, simply replace its cost by zero.
If $l_p$ is independent, then we need only add the cost of $l_p$ to $D_\downarrow(p-1, q, x, y)$, which is easily determined in constant time.
If $l_p$ is dependent, then the previous leftmost ADL $a_{\leftarrow, i}$ with only one endpoint in ${\Downarrow_B}(p-1, q)$ now lies completely within ${\Downarrow_B}(p, q)$.
Given that $a_{\leftarrow, i+1}$ is positioned at $x$, we need to determine two values: the minimum cost of all blocks in ${\Downarrow_B}(p-1, q)$ plus the cost of $a_{\leftarrow, i}$, and the minimum cost of the subtree $T_{B'}$ associated with $l_p$ plus the cost of $l_p$.
These two parts need to agree on where they place the left endpoint $\chi$ of $a_{\leftarrow, i}$.
More precisely, these costs are:
\begin{align}
    &\min_{x'} (D_\downarrow(p-1, q, x', y) + \lambda(a_{\leftarrow, i}, \chi, x')) \text{, and}\label{eq:D'}\\
    &\min_{k'} ( P(B', k', x, \chi) + \lambda(l_p, k', q)) \text{, where DL } l_p \text{ connects } B \text{ to a child bar } B'.%
    \label{eq:P'}
\end{align}
Cost $D_\downarrow(p, q, x, y)$ is then the minimum over $\chi$ of the sum of these costs.
See the right of \autoref{fig:dependent-non-adjacent-subgraph-forest:2:app} for an illustration.
The dynamic program for $D_\uparrow$ is similar.

\subparagraph{Proving the theorem.}
To compute the minimum vertical link length for $G$, we process each tree in $F$ via a post-order traversal.
Consider a bar $B$ and assume all $P$ values of child bars have been computed and stored.
Let $d$ be the degree of $B$ and $\Delta$ be the maximum degree of a bar in graph $G$.
First, we compute $D_\downarrow$ efficiently as follows, for increasing values of $p$ and $q$:
\begin{enumerate}
\item Precompute Cost~\ref{eq:D'} in $O(\Delta)$ time for each value of $y$ and $\chi$. This takes $O(\Delta^3)$ time for one pair $(p,q)$.
\item Precompute Cost~\ref{eq:P'} in $O(\Delta)$ time for each value of $x$ and $\chi$. This takes $O(\Delta^3)$ time for one pair $(p,q)$.
\item Compute $D_\downarrow$ in $O(\Delta)$ time for each value of $x$ and $y$, using two look-ups in our precomputed costs for each possible $\chi$. This takes $O(\Delta^3)$ time for one pair $(p,q)$.
\end{enumerate}
With $O(d^2)$ pairs $(p,q)$, we hence spend $O(\Delta^3d^2)$ time per bar to determine $D_\downarrow$.
After treating $D_\uparrow$ analogously, we compute each of the $O(\Delta^2d)$ entries of $P$ in $O(\Delta^2)$ time, and thus in $O(\Delta^4d)$ time per bar.
As processing a single bar takes $O(\Delta^4 d)$ time, the total running time becomes $O(\sum \Delta^4d) = O(n^4m)$.

\section{General case}
\label{sec:fpt}

We now investigate the case of a general graph $G$, that is, without placing further assumptions on the structure of the dependent links. Although the complexity of the problem remains open, a simple reduction from 3-Partition shows that a generalized version---in which bars have multiple unlinked blocks that can be ordered arbitrarily---is strongly NP-hard, as shown in \autoref{thm:nphardmultivertex} at the end of this section.

First, we prove the main theorem of this section (\autoref{thm:generalCaseFPT}), a fixed-parameter tractability result that uses two structural parameters:

\begin{description}
\item[$\Delta$:] the maximum degree of a bar in the full graph $G$;
\item[$\delta$:] the maximum degree of a bar in the subgraph of dependent links.
\end{description}

As $\delta \leq \Delta \leq n-1$, \autoref{thm:generalCaseFPT} implies that our problem is fixed-parameter tractable in $\Delta$, and admits a polynomial-time algorithm when only $\delta = O(1)$, but $\Delta$ may be large.

\begin{restatable}{theorem}{generalCaseFPT}\label{thm:generalCaseFPT}
Given $G = (V, E, w)$, the vertical link length can be minimized in $O(\Delta^{3\delta} \delta  n)$ time.
\end{restatable}

Our algorithm uses a rooted nice tree decomposition. 
These specific decompositions are generally useful for presenting dynamic programs, and are standard concepts in literature; see, for example, \cite{TreeDecomp}. We present their main aspects below, before presenting our results.

\subparagraph{Rooted nice tree decompositions.}
Let $G' = (V,E')$ denote the subgraph of dependent links. 
A \emph{rooted tree decomposition} of $G'$ is a rooted tree $T$ in which each node\footnote{To avoid ambiguity, we use nodes to refer to the vertices of $T$.} $u$ of $T$ has an associated \emph{bag} $X_u$ of bars from $V$, each link from $E'$ has at least one bag containing both its endpoints, and every bar of $V$ occurs only in the bags of a nonempty, connected subtree of $T$.
Such a decomposition is \emph{nice} if the root bag is empty and each node $u$ is one of the following types:

\begin{description}
    \item[Leaf:] $u$ has no children and $X_u = \emptyset$.
    \item[Forget:] $u$ has one child $v$ with $X_u \subset X_v$ and $|X_u|= |X_v| -1$.
    \item[Introduce:] $u$ has one child $v$ with $X_u \supset X_v$ and $|X_u|= |X_v| + 1$.
    \item[Join:] $u$ has two children $v$ and $w$ with $X_u = X_v = X_w$.
\end{description}

Due to its outerplanarity (\autoref{obs:noDLcross}), the treewidth of $G'$---the minimal size of the largest bag in such a decomposition, minus one---is at most $2$ \cite{outerplanarTW}. Thus, we may compute a nice rooted tree decomposition of $O(n)$ nodes with at most $3$ bars per bag in $O(n)$ time \cite{cygan2015parameterized}. In the remainder, we omit the qualifications nice and rooted, and denote this tree decomposition as a triple $(T,X,r)$.

\subparagraph{Dynamic program.} 
To minimize the vertical link length in a linked bar chart of~$G$, we run a dynamic program as a post-order traversal on the tree decomposition $(T, X,r)$ of $G'$. Given a node $u \in T$, we denote the set of nodes in the subtree rooted at $u$ as $T_u$.
We say that a bar is in the \emph{past} of $u$ if it is contained in $\bigcup_{v \in T_u} X_v \setminus X_u$, that is, in the bag of a descendant of $u$, but not in that of $u$ itself. The past of $u$ is then the set of all bars in the past of $u$; see \autoref{fig:TreeDecomp}.

\begin{figure}[t]
    \includegraphics[page=3]{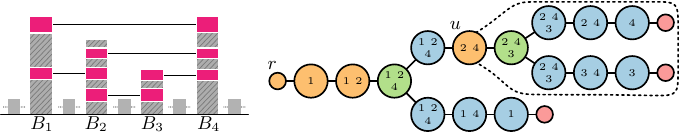}
    \centering
    \hfill
    \caption{Left: schematic of a linked bar chart, highlighting the subgraph $G'$ of dependent links on bars $B_1, B_2, B_3, B_4$. Gray blocks in between these bars represent sequences of bars with only independent blocks. Hatched regions in gray bars indicate the remaining independent blocks.
    Right: a tree decomposition of $G'$, with numbers referring to bar indices. Each node has a type: forget (orange), join (green), introduce (blue), or leaf (red). At forget node $u$, $X_u = \{B_2,B_4\}$ and bar $B_3$ is forgotten. As $B_3$ is the only bar in subtree $T_u$ (dashed outline) that is not in $X_u$, the past of $u$ is precisely~$\{ B_3 \}$. 
    }
    \label{fig:TreeDecomp}
\end{figure}

For a node $u \in T$, our dynamic program $P^*(u, \ldots)$ represents the minimum cost of all independent blocks in the bars of the past of $u$ and of all dependent links with at least one of its endpoints in the past of $u$.    
The additional parameters of $P^*(u,\ldots)$ are the way in which the endpoints of the dependent links in each of the bars in $X_u$ are positioned: we refer to this as the \emph{state} of a bar. For a node with a bag of three bars, we symbolically represent them using $S_1, S_2, S_3$; for fewer bars in its bag, we number accordingly.

The root, having an empty bag, has no further parameters and all bars are in its past, so $P^*(r)$ represents the minimal vertical link length.

\begin{restatable}{lemma}{generalCaseEntryReadTime}\label{lem:generalCaseEntryReadTime}
A node in $(T,X,r)$ has at most $O(\Delta^{3\delta})$ states.
\end{restatable}
\begin{proof}
The state $S_i$ of a bar $B_i$ effectively describes its position, for each DL incident to $B_i$: the number of rightward blocks below it for a leftward block and vice versa.
Such a state is valid, only if each block maps to a unique index in the total stacking order. As such, $S_i$ describes an injective function from at most $\delta$ dependent links, to indices in $\{ 1,\ldots, \Delta\}$, subject to ordering constraints. Hence, there are at most $\frac{\Delta!}{(\Delta-\delta)!} = O(\Delta^\delta)$ 
possible states for one bar. As a node in the tree decomposition has at most three bars in its bag, the number of possible states for one node is hence at most $O(\Delta^{3\delta})$.
\end{proof}

Although we symbolically represent these states with at most three parameters, we must observe that these are $\delta$-dimensional objects and, as such, reading an entry of $P^*(u, \ldots)$ takes $O(\delta)$ time.

We compute $P^*(u,\ldots)$ by treating each node type separately, as shown by \autoref{lem:processLeafJoinIntro} and \autoref{lem:processForget} below. In particular, one node takes $O(\Delta^{3\delta} \delta)$ time. As there are $O(n)$ nodes to process, the total time required to compute $P^*(r)$ is $O(\Delta^{3\delta} \delta n)$, which proves \autoref{thm:generalCaseFPT}. 

\begin{restatable}{lemma}{processLeafJoinIntro}\label{lem:processLeafJoinIntro}
Processing a join, introduce, or leaf node takes $O(\Delta^{3\delta} \delta)$ time.
\end{restatable}
\begin{proof}
    Each node type can be processed separately as detailed below.
    
\subparagraph{Leaf:}
    The past of $u$, and $X_u$ are empty: $P^*(u)$ is trivially 0, and takes $O(1)$ time to~compute.
    
\subparagraph{Join:}
    Let $v,w$ denote the two children of $u$ in $T$. By definition, $X_u = X_v = X_w$. The pasts of the children are disjoint: a bar that is both in $T_v$ and $T_w$ must be in $X_u$. Thus, the cost for $P^*(u,\ldots)$ is simply the sum of the costs of the children.
    Looking up both values takes $O(\delta)$ time; doing so for all $O(\Delta^{3\delta})$ potential states of $P^*(u, \ldots)$ takes  $O(\Delta^{3\delta} \delta)$ time.

\subparagraph{Introduce:}
    Node $u$ has precisely one bar in its bag that is not present in the bag of its child $v$.
    This bar therefore is not part of the past of $u$, nor can it have a dependent link to a bar in the past, due to the requirements of a tree decomposition. Therefore, the cost for $P^*(u, \ldots)$ is identical to that of its child node, given the states of the two bars that they have in common. 
    For example, if $|X_u| = 3$ and $S_\textrm{i}$ is the state of the introduced bar, this means that we compute 
    $P^*(u, S_1, S_2, S_\textrm{i}) = P^*(v, S_1, S_2)$
    Looking up this value takes $O(\delta)$ time; doing so far all $O(\Delta^{3\delta})$ potential states of $P^*(u, \ldots)$ takes  $O(\Delta^{3\delta} \delta)$ time.
\end{proof}

\begin{restatable}{lemma}{processForget}\label{lem:processForget}
Processing a forget node takes $O(\Delta^{3\delta} \delta)$ time.
\end{restatable}

\begin{proof}
    In a forget node $u$, some bar $B_\textrm{f}$ has become part of the past: it is not in $X_u$ but it is in the bag $X_v$ of the child node $v$.
    To appropriately represent the cost of the past, we must therefore include two components: the cost of the independent blocks of $B_\textrm{f}$ and the cost of the dependent links between $B_\textrm{f}$ and the bars in $X_u$ (and thus also in $X_v$).
    We iterate over all possible states $\mathcal{S}_\textrm{f}$ of $B_\textrm{f}$: together with the states prescribed for the remaining bars, this fully determines the cost of their dependent links (if present). Furthermore, it prescribes precisely between which dependent links each independent block is to be placed. 

    \newcommand{\IL}{\mathit{IL}}
    \newcommand{\DL}{\mathit{DL}}
    In the remainder, we assume that $X_u$ has two bars, $B_1$ and $B_2$---three is not possible as $X_v$ has one more bar. We can then express our dynamic program as 
    \[
        P^*(u, S_1, S_2) = \min_{ S \in \mathcal{S}_\textrm{f}} P^*(v, S_1, S_2, S) + \IL(S) + \DL(S, S_1) + \DL(S, S_2),
    \]
    where $\IL(S)$ indicates the minimal cost of the independent blocks of $B_\textrm{f}$ given the state $S$ of its dependent links, and $\DL(S,S_i)$ the cost of the dependent link between the bars of which the states are given. 
    
    Computing $\DL(S,S_i)$ is comparatively straightforward. If a link $e$ between $B_\textrm{f}$ and $B_i$ exists, then $\DL(S,S_i) = \lambda(e,p,q)$ where $p,q$ are the placements of the blocks of $e$ prescribed by $S$ and $S_i$; otherwise, $\DL(S,S_i) = 0$. Determining the existence of a link and the placements can trivially be done in $O(\delta)$ time.
    
    To compute $\IL(S)$, we first observe that the placement of the dependent blocks prescribed by $S$, partitions the independent blocks into sets: one set between each pair of consecutively placed dependent blocks, one before the first dependent block and one after the last dependent block.
    One such a set contains exactly a consecutive subset, in order, of $L_\textrm{f}$ and of $R_\textrm{f}$. As such, we can characterize the subset through four indices: $i,p$ for $L_\textrm{f}$ and $j,q$ for $R_\textrm{f}$. That is, the set is $\{ l_i, \ldots, l_p \} \cup \{ r_j, \ldots, r_q\}$. 

\begin{figure}[t]
\centering
    \includegraphics[page=2]{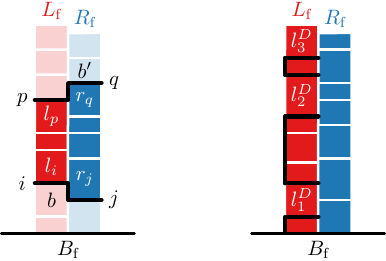}
    \caption{Left: $D_{i, j}(p, q)$ computes the minimum cost of the subset $\{l_i,\ldots ,l_p\} \cup \{r_j, \ldots, r_q\}$ between dependent blocks $b$ and $b'$. Right: the dependent links $l_1^D, l_2^D, l_3^D \in L_f$ partition the independent links in $L_f$ into three disjoint subsequences. For each such subsequence, we compute $D_{i, j}$, with $l_i$ being the lowest in a subsequence, for all $|R_f|$ possible values of $j$. 
    }
    \label{fig:AppILinFPT}
\end{figure}

    Note that each such a set can be ordered independently, and that the cost of such a set does not depend on the state $S$ either, beyond it determining which sets to consider. So, we can precompute, once per forget node, the costs of such sets.

    Let $D_{i,j}(p,q)$ represent the minimum cost of stacking blocks $\{l_{i}, \ldots,l_p\}\subseteq L_\textrm{f}$ and $\{r_{j}, \ldots,r_q\}\subseteq R_\textrm{f}$, on top of the lower blocks of both sets---note that the cost of these lower blocks is not included and their total height is fixed. 
    We can compute $D_{i,j}$ using a dynamic program akin to that of \autoref{lem:noDL}, with $i,j$ defining the base case (see \autoref{fig:AppILinFPT}, left). 
    We only need values of $p$ and $q$ such that the corresponding sets do not include dependent links; their maximal values are easily obtained from $L_\textrm{f}$ and $R_\textrm{f}$. Moreover, we only need $D_{i,j}$ where at least $i$ or $j$ is the first link after a dependent link, or the overall first link, in $L_\textrm{f}$ or $R_\textrm{f}$, respectively. For $D_{i,j}$ with $i$ being such a first link, let $p_i$ denote the largest value such that $\{l_i, \ldots, l_{p_i}\} $ contains no dependent link: these subsequences are disjoint for all such $i$ and correspond to $L_\textrm{f}$, partitioned by the dependent links: $\sum_i (p_i - i + 1) = \Delta - \delta$. With $|R_\textrm{f}|$ positions for $j$, and a dynamic program that takes $O((p_i-i+1)|R_\textrm{f}|)$ time to compute (see \autoref{fig:AppILinFPT}, right), the total time spent for these instances of $D_{i,j}$ is $O(\sum_i (p-i+1) |R_\textrm{f}|^2) = O(\Delta^3)$. Analogously, instances $D_{i,j}$ with $j$ just after a dependent link takes $O(\Delta^3)$ time.
    The total preprocessing time for a forget node is thus $O(\Delta^3)$.

    With $D_{i,j}$ available, we can compute $\IL(S)$, after ordering the placements of the dependent links in $O(\delta)$ time via a merge procedure,
    traversing all $O(\delta)$ consecutive pairs and looking up the minimum cost of the intermediate independent blocks using $D$ in $O(1)$ time. 

    Computing the result one state $S \in \mathcal{S}_\textrm{f}$ takes $O(\delta)$ time; processing one state of $P^*(u,\ldots)$ takes $O(\Delta^\delta \delta)$. As a forget node has at most two bars, it has $O(\Delta^{2\delta})$ states. Thus, the total time to process a forget node is $O(\Delta^{3\delta} \delta + \Delta^3) = O(\Delta^{3\delta}\delta)$.
\end{proof}

\begin{restatable}{theorem}{NPHardMultiVertex}\label{thm:nphardmultivertex}
Minimizing vertical link length for a fixed bar order is strongly NP-hard, when vertices have multiple unlinked blocks and these may be ordered arbitrarily throughout the linked blocks.
\end{restatable}

\begin{figure}[b]
\centering
\includegraphics{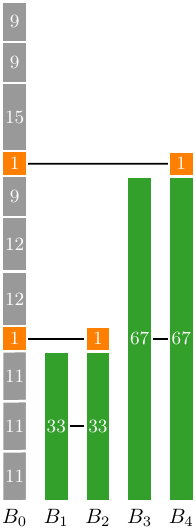}
\caption{Example construction for $S = \{ 9,9,9, 11,11,11, 12,12,15 \}$, requiring $k = 3$ triples of $T = 33$. It uses $5$ bars, with $2$ links from $B_0$ (orange) and 2 links to fix the vertical position at the other bars (green). For compactness, vertical height has been compressed, and the linked blocks of $B_0$ have been exaggerated for clarity.} 
\label{fig:3partNPhard}
\end{figure}


\begin{proof}
We prove a reduction from 3-Partition, which is strongly NP-hard \cite{NP-Completeness}. In this problem, we are given a multiset of positive integers $S = \{s_1, \ldots, s_n \}$ and the question is whether $S$ can be partitioned into $k = n/3$ triples such that each triple sums to the same value. Let $T = \sum_{s \in S} s / k$ denote the total value of a triple in a solution. The problem remains strongly NP-hard, even when every integer in $S$ must strictly lie between $T/4$ and $T/2$: that is, one can only obtain a sum of $T$ with exactly three integers.

Our reduction uses $2k - 1$ bars $B_{0}, \ldots, B_{2k-2}$ (see \autoref{fig:3partNPhard}). Bar $B_0$ has $n$ unlinked blocks, with weights matching exactly those of $S$. Moreover, it has $k-1$ linked blocks of weight $1$, linking to the even-numbered bars $B_{2}, B_{4}, \ldots, B_{2k-2}$. A bar $B_{2i}$ for $i \geq 1$ further has one more linked block of weight $i \cdot T + (i-1)$, connecting it to bar $B_{2i-1}$.

Below, we argue that this linked bar chart can be ordered to achieve vertical link length zero, if and only if $S$ admits a 3-partition. 

By construction, bars $B_{i}$ with $i \geq 1$ have no choice in their block ordering, as there are no unlinked blocks and all links leave the bar on one side and their blocks must, hence, respect the bar ordering. This also implies that the links between $B_{2i-1}$ and $B_{2i}$ for $i \geq 1$ have zero vertical length, as both blocks are placed at the bottom. The other block at $B_{2i}$ has its vertical center at $i \cdot T + (i-1) + 0.5$. This value readily exceeds the maximum height of all bars between $B_{0}$ and $B_{2i}$.

For bar $B_0$, the order of the linked blocks is also fixed, as all links are towards the right. Hence, the only choice is how to arrange the unlinked blocks in between the linked blocks.

If $S$ admits a 3-partition, then we place exactly one triplet at the bottom and top, and one triplet in between each pair of linked blocks. As there are $i$ triplets before the linked block to bar $B_{2i}$, and $(i-1)$ earlier linked blocks, this puts the block to bar $B_{2i}$ such that its center is at height $i \cdot T + (i-1) + 0.5$, matching the vertical center at bar $B_{2i}$. Hence, all links incur zero vertical link length.

If the linked bar chart admits a solution with zero vertical link length, then each of the linked blocks in $B_0$ is placed at a height matching the other end point. As such, the value before the block to bar $B_{2}$ must be $T$ and the value before $B_{2i}$ for $i > 1$ must be $i\cdot T+(i-1)$: thus a total height of $T$ was created between the block for $B_{2i}$ and the linked block to $B_{2i-2}$. By the restriction on the values in $S$, we can only create a subset that sums to $T$ with a triple, and hence the ordering induces a 3-partition of $S$.
\end{proof}

\section{Conclusion}

We studied algorithmically minimizing vertical link length for linked bar charts. This problem provides a new dimension to one-page book embeddings, by considering how the blocks (vertex and edge weights) at each bar (vertex) are to be stacked. For a fixed bar order, we described polynomial-time algorithms for special cases on the subgraph of dependent links. The general case is fixed-parameter tractable, parameterized by the maximum degree of a bar; the algorithm is polynomial when the degree in the dependent subgraph is bounded by a constant.

\subparagraph{Future work.}
The primary open problem is deciding whether minimizing vertical link length for a fixed bar order is NP-hard. Though we show a generalized version to be NP-hard, this reduction relies mostly on plurality of vertex weights. Our FPT results imply that, if a reduction is to be shown, this must rely on an instance with bars having many dependent links, and that is well connected to avoid being close to a tree.

Furthermore, we assumed a fixed bar order: what are the algorithmic possibilities for optimizing linked bar charts when we can control both the bar order as well as the stacking order of the blocks?

\subparagraph{Generalizations.}
There are several natural generalizations to our problem.

\begin{description}
    \item[Asymmetric links:] The original design assumes that the blocks on both ends of a link are of equal height. However, one could also use a design where these sizes are not the same, for example, to communicate an expected distribution. We observe that our algorithms do not rely on linked blocks being the same size, and as such, generalize to solve this case.
    \item[Other measures:] Beyond vertical link length, other quality measures could be considered. Our results straightforwardly apply when the measure permits decomposing its cost for independent links and computing the cost of a link takes $O(1)$ given its placed endpoints; for example, bend minimization.    
    \item[Directed graphs and multigraphs:] The original design assumes an undirected graph, with at most one link between two bars. A natural generalization would be to consider directed or multiple links between two bars.  In the context of shared quantities, it may be relevant to show different sizes, for example, to indicate communication on different platforms, or to differentiate between sending and receiving messages. 
    This opens up new questions, depending on whether the multiple links between two bars can be reordered, and whether they should be placed as one contiguous chunk within a bar. 
    \item[Hypergraphs:] A natural extension is also to go beyond regular graphs and instead consider a hypergraph to describe the cross-category values, linking multiple bars at once; representing, for example, communication in a group chat. Again, this generalization gives rise to new design questions that need to be answered before the algorithmic problem can be studied. The foremost question is how to represent a single link between three or more bars. 
\end{description}



\bibliography{arxiv}

\end{document}